\documentclass[a4paper,UKenglish,cleveref, autoref, thm-restate]{lipics-v2021}
\newif\iflong\longtrue

\iflong
\hideLIPIcs  
\fi

\title{Worst-Case and Smoothed Analysis of the Hartigan--Wong Method for k-Means Clustering} 

\titlerunning{Worst-Case and Smoothed Analysis of the Hartigan--Wong Method} 

\author{Bodo Manthey}{Faculty of Electrical Engineering, Mathematics, and Computer Science, University of Twente, The Netherlands \and \url{https://people.utwente.nl/b.manthey}}{b.manthey@utwente.nl}{https://orcid.org/0000-0001-6278-5059}{}
\author{Jesse van Rhijn}{Faculty of Electrical Engineering, Mathematics, and Computer Science, University of Twente, The Netherlands \and \url{https://people.utwente.nl/j.vanrhijn}}{j.vanrhijn@utwente.nl}{https://orcid.org/0000-0002-3416-7672}{Supported by NWO grant OCENW.KLEIN.176.}

\authorrunning{B. Manthey and J. van Rhijn}

\Copyright{Bodo Manthey and Jesse van Rhijn}



\category{} 

\iflong
\else
\relatedversiondetails{Full version}{https://arxiv.org/abs/2309.10368} 
\fi

\nolinenumbers 

\EventEditors{Olaf Beyersdorff, Mamadou Moustapha Kant\'{e}, Orna Kupferman, and Daniel Lokshtanov}
\EventNoEds{4}
\EventLongTitle{41st International Symposium on Theoretical Aspects of Computer Science (STACS 2024)}
\EventShortTitle{STACS 2024}
\EventAcronym{STACS}
\EventYear{2024}
\EventDate{March 12--14, 2024}
\EventLocation{Clermont-Ferrand, France}
\EventLogo{}
\SeriesVolume{289}
\ArticleNo{8}

\let\given\givenbase

\newcommand{\poly}{\operatorname{poly}}

\newcommand{\cm}{\operatorname{cm}}

\def\prob{\ensuremath\mathbb{P}}
\def\expect{\ensuremath\mathbb{E}}
\newcommand*\dd{\mathop{}\!\mathrm{d}}
\def\C{\ensuremath\mathcal{C}}
\def\X{\ensuremath\mathcal{X}}
\def\Y{\ensuremath\mathcal{Y}}

\newcommand{\equal}{=}

\let\originalleft\left
\let\originalright\right
\renewcommand{\left}{\mathopen{}\mathclose\bgroup\originalleft}
\renewcommand{\right}{\aftergroup\egroup\originalright}

\crefname{ineq}{inequality}{inequalities}
\creflabelformat{ineq}{#2{\upshape(#1)}#3} 

\usepackage{algorithm2e}
\usepackage{thmtools}

\usepackage{regexpatch}
\makeatletter
\xpatchcmd\thmt@restatable{%
\csname #2\@xa\endcsname\ifx\@nx#1\@nx\else[{#1}]\fi
}{%
\ifthmt@thisistheone
\csname #2\@xa\endcsname\ifx\@nx#1\@nx\else[{#1}]\fi
\else
\csname #2\@xa\endcsname[{Restated}]
\fi}{}{}
\makeatother

\usepackage{tikz}
\usepackage{tkz-graph}
\usetikzlibrary{calc}
\usetikzlibrary{fit, shapes.geometric}
\tikzstyle{vertex}=[circle, fill, inner sep=0pt, minimum size=0.1pt]

\tikzset{snake it/.style={decorate, decoration=snake}}
\tikzset{
dot/.style = {circle, fill, minimum size=#1,
              inner sep=0pt, outer sep=0pt},
dot/.default = 3pt 
}

\begin{document}

\maketitle

\begin{abstract}
    We analyze the running time of the Hartigan--Wong method, an old algorithm for the~$k$-means
    clustering problem. First, we construct an instance on the line on which
    the method can take~$2^{\Omega(n)}$ steps to converge, demonstrating that the Hartigan--Wong method
    has exponential worst-case running time even when~$k$-means
    is easy to solve.
    As this is in contrast to the empirical performance
    of the algorithm, we also analyze the running time in the
    framework of smoothed analysis. In particular,
    given an instance of~$n$ points in~$d$ dimensions, we prove that
    the expected number of iterations needed for the Hartigan--Wong method
    to terminate is bounded by~$k^{12kd}\cdot \poly(n, k, d, 1/\sigma)$
    when the points in the instance are perturbed by independent 
   ~$d$-dimensional Gaussian random variables of mean~$0$ and standard
    deviation~$\sigma$.
\end{abstract}

\section{Introduction}

Clustering is an important problem in computer science, from both
a practical and a theoretical perspective. On the practical side,
identifying clusters of similar points in large data sets has relevance
to fields ranging from physics to biology to sociology. Recent advances
in machine learning and big data have made the need for efficient clustering
algorithms even more apparent. On the theoretical side, clustering problems
continue to be a topic of research from the perspective of approximation algorithms,
heuristics, and computational geometry.

Perhaps the best-studied clustering problem is that of~$k$-means clustering.
In this problem, one is given a finite set of points~$\X \subseteq \mathbb{R}^d$
and an integer~$k$. The goal is to partition the points into~$k$ subsets, such that
the sum of squared distances of each point to the centroid of its assigned cluster,
also called its cluster center, is minimized.

Despite great effort to devise approximation algorithms for~$k$-means clustering,
the method of choice remains Lloyd's method \cite{lloydLeastSquaresQuantization1982}. This method
starts with an arbitrary choice of centers, and assigns each point to its closest
center. The centers are then moved to the centroids of each cluster. In the next
iteration, each point is again reassigned to its closest center, and the process repeats.

It is not hard to show that this process strictly decreases the objective function
whenever either a cluster center changes position, or a point is reassigned. Hence,
no clustering can show up twice during an execution of this algorithm. Since
the number of partitions of~$n$ points into~$k$ sets is at most~$k^n$, the process
must eventually terminate.

Although Lloyd's method has poor approximation performance both in theory
and in practice \cite{arthurKmeansAdvantagesCareful2007}, its speed has kept it relevant to
practitioners. This is in startling contrast to its
worst-case running time, which is exponential in the number of points
\cite{vattaniKmeansRequiresExponentially2011}.

To close the gap between theory and practice, Arthur et al.\ have shown that
Lloyd's method terminates in expected polynomial time on
perturbed point sets, by means of a smoothed analysis
\cite{arthurSmoothedAnalysisKMeans2011}. This
provides some theoretical justification for the use of Lloyd's method
in practice.

Another, less well-known heuristic for clustering is the Hartigan--Wong method 
\cite{hartiganAlgorithm136KMeans1979}.
In this method, one proceeds point-by-point. Given an arbitrary clustering,
one checks whether there exists a point that can be reassigned to a different cluster,
such that the objective function decreases.
If such a point exists, it is reassigned to this new cluster. If no such
points exist, the algorithm terminates and the clustering is declared locally optimal.

Although at first sight the Hartigan--Wong method might seem like
a simpler version of Lloyd's method, it is qualitatively different.
If Lloyd's method reassigns a point~$x$ from cluster~$i$ to cluster~$j$,
then~$x$ must be closer to the center of cluster~$j$ than to that of
cluster~$i$. In the Hartigan--Wong method, this is not true;~$x$ may be reassigned
even when there are no cluster centers closer to~$x$ than its current center.
This can be beneficial, as Telgarsky \& Vattani showed that the Hartigan--Wong method
is more powerful than Lloyd's method \cite{telgarskyHartiganMethodKmeans2010}.

To be precise, every local
optimum of the Hartigan--Wong method is also a local optimum of LLoyd's method,
while the converse
does not hold. Telgarsky \& Vattani
moreover performed computational experiments, which show that
the Hartigan--Wong method not only tends to find better clusterings than Lloyd's, but
also has a similar running time on practical instances. 
Despite these promising results, theoretical knowledge of the Hartigan--Wong method
is lacking. 

In this paper, we aim to advance our understanding of this heuristic.
Our contributions are twofold. First, we construct an instance on the line
on which the Hartigan--Wong method
can take~$2^{\Omega(n)}$ iterations to terminate. Considering that~$k$-means clustering can be solved exactly in polynomial time
in~$d = 1$, this shows that the worst-case running time of the Hartigan--Wong method
is very poor even on easy instances.
This is in contrast to Lloyd's method,
where all known non-trivial lower bounds require~$d \geq 2$.

\begin{theorem}[restate=lowerbound]\label{thm:lower_bound}
    For each~$m \in \mathbb{N}_{\geq 2}$ there exists an instance of
   ~$k$-means clustering on the line with~$n = 4m-3$ points and~$k = 2m - 1$ clusters
    on which the Hartigan--Wong method can
    take~$2^{\Omega(n)}$ iterations to converge to a local optimum.
\end{theorem}

Second, we attempt to reconcile \Cref{thm:lower_bound} with
the observed practical performance of the Hartigan--Wong method.
We perform a smoothed analysis of its running time, in which each point
in an arbitrary instance is independently perturbed by a Gaussian random
variable of variance~$\sigma^2$. 

\begin{theorem}[restate=smoothedcomplexity]\label{thm:smoothed_complexity}
    Let~$n, k, d \in \mathbb{N}$, and assume~$4kd \leq n$.
    Fix a set of~$n$ points~$\Y \subseteq [0, 1]^d$, and assume that each point
    in~$\Y$ is independently perturbed by a~$d$-dimensional Gaussian random
    variable with mean~$0$ and standard deviation~$\sigma$, yielding a new
    set of points~$\X$. Then the expected running time of the Hartigan--Wong method
    on~$\X$ is bounded by
    \[
    O\left(
                    \frac{k^{12kd+5} d^{12} n^{12.5 +\frac{1}{d}}\ln^{4.5}(nkd)}
                    {\sigma^4}
    \right) = k^{12kd} \cdot \poly(n, k, d, 1/\sigma).
    \]
\end{theorem}

Although we do not attain a polynomial smoothed running time in all
problem parameters, we note that for Lloyd's method one of the first smoothed
analyses yielded a similar~$k^{kd}\poly(n, 1/\sigma)$ bound.
This was later improved to~$\poly(n, k, d, 1/\sigma)$.
We therefore regard \Cref{thm:smoothed_complexity} as a first step to 
settling the conjecture by
Telgarsky \& Vattani that the Hartigan--Wong method, like Lloyd's method, should have
polynomial smoothed running time.

We note that \Cref{thm:lower_bound} shows that there exists an instance on which 
there exists some very specific sequence of iterations that has exponential length.
In essence, this means that the exponential running time is only shown for
a very specific pivot rule for choosing which point to reassign to which cluster in
each iteration. By contrast, \Cref{thm:smoothed_complexity} holds for
\emph{any} pivot rule, not simply for any particular choice.

\section{Preliminaries and Notation}

Given vectors~$x, y \in \mathbb{R}^d$, we write~$\langle x, y \rangle$ for the standard Euclidean inner product on~$\mathbb{R}^d$, and~$\|x\| = \sqrt{\langle x, x \rangle}$ for the standard norm.

Given a set of~$k$ clusters~$\C = \{\C_1, \ldots, \C_k\}$, a configuration
of a cluster~$\C_i \in \C$ is an assignment of a set of points
to~$\C_i$. We will denote the clusters by calligraphic letters, and
their configurations by regular letters; i.e., the configuration
of~$\C_i$ will be denoted~$C_i$. This distinction is sometimes useful.
For the majority of this paper, however, we will not make this
distinction explicitly, and will refer to both a cluster and its configuration
interchangeably by regular letters.

Given a finite set of points~$S \subseteq \mathbb{R}^d$, we define the center
of mass of~$S$ as
\[
    \cm(S) = \frac{1}{|S|} \sum_{x \in S} x.
\]
With this definition, we can formally define the objective function of~$k$-means.
Let~$C = \{C_i\}_{i=1}^k$ be a partition of a finite set of points~$\X \subseteq \mathbb{R}^d$.
Then the objective function of~$k$-means is
\[
    \Phi(C) = \sum_{i=1}^k \sum_{x \in C_i} \|x - \cm(C_i)\|^2 = \sum_{i=1}^k \Phi(C_i),
\]
where we define~$\Phi(C_i) = \sum_{x \in C_i}\|x - \cm(C_i)\|^2$.
We will also refer to~$\Phi(C)$ as the potential function.

%

For both the worst-case and smoothed complexity bounds, we need
to analyze the improvement of a single iteration. Thus, we need a
simple expression for this quantity. \Cref{lemma:merge,lemma:move}
allow us to obtain such an expression. These results
were already obtained by Telgarsky \& Vattani \cite{telgarskyHartiganMethodKmeans2010}.

\begin{lemma}[Telgarsky \& Vattani \cite{telgarskyHartiganMethodKmeans2010}]\label{lemma:merge}
    Let~$S$ and~$T$ be two disjoint nonempty sets of points in~$\mathbb{R}^d$. Then
    \[ 
        \Phi(S \cup T) - \Phi(S) - \Phi(T) = \frac{|S|\cdot |T|}{|S| + |T|}
            \cdot \|\cm(S) - \cm(T)\|^2.
    \]
\end{lemma}

\begin{lemma}[Telgarsky \& Vattani \cite{telgarskyHartiganMethodKmeans2010}]\label{lemma:move}
    Let~$S$ and~$T$ be two disjoint nonempty sets of points in~$\mathbb{R}^d$
    with~$|S| > 1$. Suppose we
    move a point~$x \in S$ from~$S$ to~$T$. Then
    \[
        \Phi(S\setminus\{x\}) + \Phi(T \cup \{x\}) - \Phi(T) - \Phi(S)
            = \frac{|T|}{|T| + 1} \|\cm(T) - x\|^2 - \frac{|S|}{|S| - 1}
                \|\cm(S) - x\|^2.
    \]
\end{lemma}


Let~$C$ be some clustering of~$\X$. Suppose in some iteration of the Hartigan--Wong method,
we move~$x \in C_i$
to~$C_j$. Let the gain of this iteration be denoted~$\Delta_x(C_i, C_j)$.
Then \Cref{lemma:move} tells us that
\[
    \Delta_x(C_i, C_j) = \frac{|C_i|}{|C_i|-1} \|x - \cm(C_i)\|^2 - \frac{|C_j|}{|C_j| + 1} \|x - \cm(C_j)\|^2.
\]

At first sight, it seems like \Cref{lemma:move} leaves open the possibility
that a cluster is left empty. The following lemma shows that this can never happen.

\begin{lemma}\label{lemma:never_empty}
    No iteration can leave a cluster empty. 
\end{lemma}

\begin{proof}
    Suppose before an iteration,~$C_i = \{x\}$ for some~$x \in X$,
    and after the iteration~$C_i' = \emptyset$ and~$C_j' = C_j \cup \{x\}$,
    i.e.~$x$ is moved from cluster~$i$ to cluster~$j$. The gain of this
    iteration is then (\Cref{lemma:merge})
    \[
        \Phi(C_i) + \Phi(C_j) - \Phi(\emptyset) - \Phi(C_j \cup \{x\})
            = \Phi(C_j) - \Phi(C_j \cup \{x\}) = - \frac{|C_j|}{|C_j| + 1} \|x - \cm(C_j)\|^2 \leq 0,
    \]
    since~$\cm(C_i) = x$ and~$\Phi(\emptyset) = 0$. Since every iteration must
    improve the clustering, this concludes the proof.
\end{proof}

\section{Exponential Lower Bound}

In this section, we construct a family of~$k$-means instances on the line
on which the Hartigan--Wong method can take an exponential number of iterations before
reaching a local optimum. To be precise, we prove the following theorem.

\lowerbound*

The construction we employ is similar to the construction used by
Vattani for Lloyd's method \cite{vattaniKmeansRequiresExponentially2011}. However,
the Hartigan--Wong method only reassigns a single point in each iteration,
and we are free to choose which point we reassign. Moreover, we are
even free to choose which cluster we move a point to if there are
multiple options. This allows us to simplify the construction
and embed it in a single dimension, rather than the plane used
by Vattani.

We define a set of~$m$ gadgets~$G_i$,~$i \in \{0, \ldots, m-1\}$. Each gadget except for the ``leaf'' gadget~$G_0$ consists of four 
points, and has two clusters~$G_i(\C_0)$ and~$G_i(\C_1)$ associated with it.
Moreover, each gadget except~$G_0$ has three distinguished states, called
``morning'', ``afternoon'', and ``asleep''. The leaf gadget
only has two states, ``awake'' and ``asleep ''.

During the morning state, a gadget~$G_i$ watches~$G_{i-1}$.
If~$G_{i-1}$ falls asleep, then it is awoken by~$G_i$; this is achieved by moving
a point of~$G_i$ to one of the clusters of~$G_{i-1}$. This allows~$G_{i-1}$ to
perform a sequence of iterations, which ends with~$G_{i-1}$ back in its morning state.

Meanwhile,~$G_{i}$ performs a sequence of iterations that transition it to
its afternoon state. During the afternoon state, it once more watches~$G_{i-1}$.
When the latter falls asleep,~$G_i$ once again wakes~$G_{i-1}$, and transitions
itself to its asleep state.

The leaf gadget~$G_0$, as it does not watch any gadgets, only ever awakens and immediately falls
asleep again.

We end the sequence of iterations once gadget~$m-1$ falls asleep. Observe that
with this construction,~$G_i$ falls asleep twice as often as~$G_{i+1}$. With the condition that~$G_{m-1}$ falls asleep once,
we obtain a sequence of at least~$2^{m-1}$ iterations. With~$n = 4m-3$,
this yields \Cref{thm:lower_bound}.

For space reasons, we only describe the instance and the exponential-length
sequence here. The proof that this sequence is improving, which completes
the proof of \Cref{thm:lower_bound}, is deferred to the 
\iflong appendix.
\else
full version.
\fi

\subsection{Formal Construction}\label{sec:lower_bound_sequence}

We now give a detailed construction of a unit gadget,~$G$. All gadgets
except for~$G_0$ are scaled and translated versions of~$G$.
The unit gadget is a tuple~$G = (S, \C_0, \C_1)$, where~$S = \{a, b, p, q\} \subseteq \mathbb{R}$, and~$\C_0$ and~$\C_1$
are two clusters. The positions of the points in~$S$ are given in \Cref{table:unit_gadget_points}.
In addition, the gadget is depicted schematically
in \Cref{fig:morning_afternoon,fig:wakeup}.
Note that the relative positions of the points
in these figures do not correspond to \Cref{table:unit_gadget_points},
but are chosen for visual clarity.

\begin{table}[ht!]
\centering
\begin{tabular}{l|l|l|l|l|l|l|l}
    Point &~$a$ &~$b$ &~$p$ &~$q$ &~$f$ &~$t_0$ \\ \hline
    Position & 9 & 6 & 5 & 13 & 0 & 8
\end{tabular}
\caption{Positions of the points in~$S(G)$, the leaf point~$f$,
and the translation vector~$t_0$ between gadgets~$G_1$ and~$G_2$.\label{table:unit_gadget_points}}
\end{table}

We remark that the points in \Cref{table:unit_gadget_points} are not simply chosen by
trial-and-error. As will be explained shortly, we can obtain from our construction a series
of inequalities that must be satisfied by the points in~$S$. We then obtained
these points by solving the model
\begin{align*}
    \min &\quad a^2 + b^2 + p^2 + q^2 + f^2 + t_0^2 \\
    \text{s.t.} &\quad \text{each move decreases the clustering cost,} \\ 
                &\quad a, b, p, q, f, t_0 \in \mathbb{Z}
\end{align*}
using Gurobi \cite{gurobioptimizationllcGurobiOptimizerReference2023}.
The first constraint here amounts to satisfying a series
of inequalities of the form~$\Delta_{x}(A, B) > 0$ for~$x \in S(G)$ and~$A, B$ subsets of the points in a gadget and its neighboring
gadgets. For space reasons, we defer their derivation and
verification to the 
\iflong
appendix.
\else
full version.
\fi

The objective function here is purely chosen so that
Gurobi prefers to choose small integers in the solution.

To construct~$G_i$ from the unit gadget (for~$i \geq 1)$, we scale the unit gadget by a factor~$5^{i-1}$,
and translate it by~$t_i = \sum_{j=0}^{i-1} 5^j t_0$, where~$t_0 = 8$. Since
each gadget only ever exchanges points with its neighbors in the sequence we are about to
construct, it will suffice in proving \Cref{thm:lower_bound} to consider only
iterations involving~$G_i$,~$G_{i-1}$ and~$G_{i+1}$ for some fixed~$i > 2$.
For the leaf gadget, we simply have~$G_0 = (S_0, \C_0)$, where~$S_0 = \{f\} = \{0\}$.

Before we go on to construct an improving sequence of exponential length,
we define the earlier-mentioned states. For ease of notation, we will refer
to the points of~$G_i$ as~$a_i$,~$b_i$, and so on, and to the clusters
of~$G_i$ as~$\C_0(G_i)$ and~$\C_1(G_i)$. Then we say the state of~$G_{i > 0}$ is:

\begin{itemize}
    \item asleep, if~$C_0(G_i) = \{b_i\}$ and~$C_1(G_i) = \{a_i, q_i\}$ (in this
    state,~$p_i$ is in some cluster of~$G_{i-1}$);
    \item morning, if~$C_0(G_i) = \{p_i, q_i, b_i\}$ and~$C_1(G_i) = \{a_i\}$;
    \item afternoon, if~$C_0(G_i) = \{b_i\}$ and~$C_1(G_i) = \{p_i, q_i, a_i\}$.
\end{itemize}

For the leaf gadget, we say its state is:

\begin{itemize}
    \item asleep, if~$C_0(G_0) = \{f\}$;
    \item awake, otherwise.
\end{itemize}

We now explicitly determine a sequence of iterations of exponential length.
In the proof of \Cref{thm:lower_bound}, we show that
this sequence is improving. To analyze the sequence, we consider
the perspective of~$G_{i}$ as it wakes up~$G_{i-1}$ and falls asleep; and then
as it is awoken by~$G_{i+1}$. We first consider only the case that~$G_{i-1} \neq G_0$. See \Cref{fig:morning_afternoon} and
\Cref{fig:wakeup} for a schematic depiction of the sequence described below.

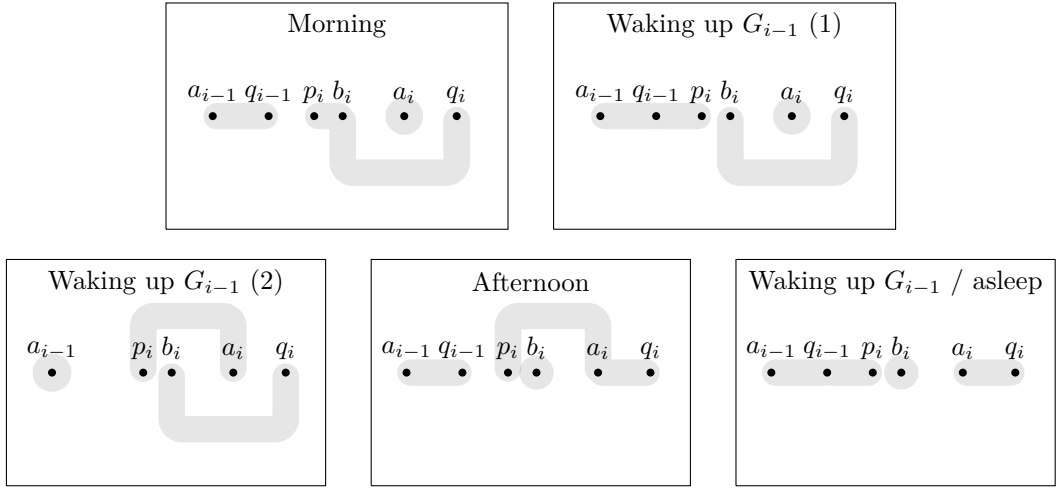
\begin{figure}
    \centering
    \begin{tikzpicture}[scale=1.5]
        \node[] at (1, 0.8) {Morning};
        \draw (-0.5, -1) -- (-0.5, 1) -- (2.5, 1) -- (2.5, -1) -- (-0.5, -1); 
    
        \node[dot, label=above:$a_i$] (a1) at (1.59, 0) {};
        \node[dot, label=above:$b_i$] (b1) at (1.05, 0) {};
        \node[dot, label=above:$p_i$] (p1) at (0.8, 0) {};
        \node[dot, label=above:$q_i$] (q1) at (2.05, 0) {};
        
        \node[dot, label=above:$a_{i-1}$] (a0) at (-0.09, 0) {};
        \node[dot, label=above:$q_{i-1}$] (q0) at (0.4, 0) {};

        \coordinate (pbelow) at (0.8, -0.5) {};
        \coordinate (bbelow) at (1.05, -0.5) {};
        \coordinate (qbelow) at (2.05, -0.5) {};

        \draw[line cap=round, rounded corners, line width=10pt, opacity=0.1] (p1) -- (b1) -- (bbelow) -- (qbelow) -- (q1);
        
        \node[ellipse, draw, fill, opacity=0.1, fit=(a1)] {};
        
        \draw[line cap=round, rounded corners, line width=10pt, opacity=0.1] (a0) -- (q0);
        
    \end{tikzpicture}
    \hspace{1em}
    \begin{tikzpicture}[scale=1.5]
        \node[] at (1, 0.8) {Waking up~$G_{i-1}$ (1)};
        \draw (-0.5, -1) -- (-0.5, 1) -- (2.5, 1) -- (2.5, -1) -- (-0.5, -1); 
        \node[dot, label=above:$a_i$] (a1) at (1.59, 0) {};
        \node[dot, label=above:$b_i$] (b1) at (1.05, 0) {};
        \node[dot, label=above:$p_i$] (p1) at (0.8, 0) {};
        \node[dot, label=above:$q_i$] (q1) at (2.05, 0) {};
        
        \node[dot, label=above:$a_{i-1}$] (a0) at (-0.09, 0) {};
        \node[dot, label=above:$q_{i-1}$] (q0) at (0.4, 0) {};

        \coordinate (pbelow) at (0.8, -0.5) {};
        \coordinate (bbelow) at (1.05, -0.5) {};
        \coordinate (qbelow) at (2.05, -0.5) {};

        \draw[line cap=round, rounded corners, line width=10pt, opacity=0.1] (b1) -- (bbelow) -- (qbelow) -- (q1);
        \node[ellipse, draw, fill, opacity=0.1, fit=(a1)] {};
        \draw[line cap=round, rounded corners, line width=10pt, opacity=0.1] (a0) -- (q0) -- (p1);
        
    \end{tikzpicture}
    \\ \vspace{1em}
    \begin{tikzpicture}[scale=1.5]
        \node[] at (1, 0.8) {Waking up~$G_{i-1}$ (2)};
        \draw (-0.4, -1) -- (-0.4, 1) -- (2.4, 1) -- (2.4, -1) -- (-0.4, -1); 
        \node[dot, label=above:$a_i$] (a1) at (1.59, 0) {};
        \node[dot, label=above:$b_i$] (b1) at (1.05, 0) {};
        \node[dot, label=above:$p_i$] (p1) at (0.8, 0) {};
        \node[dot, label=above:$q_i$] (q1) at (2.05, 0) {};
        
        \node[dot, label=above:$a_{i-1}$] (a0) at (0, 0) {};

        \coordinate (pbelow) at (0.8, -0.5) {};
        \coordinate (bbelow) at (1.05, -0.5) {};
        \coordinate (qbelow) at (2.05, -0.5) {};
        \coordinate (abelow) at (1.59, -0.5) {};
        
        \coordinate (pabove) at (0.8, 0.5) {};
        \coordinate (aabove) at (1.59, 0.5) {};

        \draw[line cap=round, rounded corners, line width=10pt, opacity=0.1] (b1) -- (bbelow) -- (qbelow) -- (q1);
        
        \draw[line cap=round, rounded corners, line width=10pt, opacity=0.1] (a1) -- (aabove) -- (pabove) -- (p1);
        
        \node[ellipse, draw, thick, fill, opacity=0.1, fit=(a0)] {};
        
    \end{tikzpicture}
    \hspace{1em}
    \begin{tikzpicture}[scale=1.5]
        \node[] at (1, 0.8) {Afternoon};
        \draw (-0.4, -1) -- (-0.4, 1) -- (2.4, 1) -- (2.4, -1) -- (-0.4, -1); 
        \node[dot, label=above:$a_i$] (a1) at (1.59, 0) {};
        \node[dot, label=above:$b_i$] (b1) at (1.05, 0) {};
        \node[dot, label=above:$p_i$] (p1) at (0.8, 0) {};
        \node[dot, label=above:$q_i$] (q1) at (2.05, 0) {};
        
        \node[dot, label=above:$a_{i-1}$] (a0) at (-0.09, 0) {};
        \node[dot, label=above:$q_{i-1}$] (q0) at (0.4, 0) {};

        \coordinate (pbelow) at (0.8, -0.5) {};
        \coordinate (bbelow) at (1.05, -0.5) {};
        \coordinate (qbelow) at (2.05, -0.5) {};
        \coordinate (abelow) at (1.59, -0.5) {};
        
        \coordinate (pabove) at (0.8, 0.5) {};
        \coordinate (aabove) at (1.59, 0.5) {};

        \node[ellipse, draw, fill, opacity=0.1, rotate fit = 45, fit=(b1)] {};
        
        \draw[line cap=round, rounded corners, line width=10pt, opacity=0.1] (q1) -- (a1) -- (aabove) -- (pabove) -- (p1);
        
        \draw[line cap=round, rounded corners, line width=10pt, opacity=0.1] (a0) -- (q0);
        
    \end{tikzpicture}
    \hspace{1em}
    \begin{tikzpicture}[scale=1.5]
        \node[] at (1, 0.8) {Waking up~$G_{i-1}$ / asleep};
        \draw (-0.4, -1) -- (-0.4, 1) -- (2.4, 1) -- (2.4, -1) -- (-0.4, -1); 
        \node[dot, label=above:$a_i$] (a1) at (1.59, 0) {};
        \node[dot, label=above:$b_i$] (b1) at (1.05, 0) {};
        \node[dot, label=above:$p_i$] (p1) at (0.8, 0) {};
        \node[dot, label=above:$q_i$] (q1) at (2.05, 0) {};
        
        \node[dot, label=above:$a_{i-1}$] (a0) at (-0.09, 0) {};
        \node[dot, label=above:$q_{i-1}$] (q0) at (0.4, 0) {};

        \coordinate (pbelow) at (0.8, -0.5) {};
        \coordinate (bbelow) at (1.05, -0.5) {};
        \coordinate (qbelow) at (2.05, -0.5) {};
        \coordinate (abelow) at (1.59, -0.5) {};
        
        \coordinate (pabove) at (0.8, 0.5) {};
        \coordinate (aabove) at (1.59, 0.5) {};

        \node[ellipse, draw, fill, opacity=0.1, rotate fit = 45, fit=(b1)] {};
        
        \draw[line cap=round, rounded corners, line width=10pt, opacity=0.1] (q1) -- (a1);
        
        \draw[line cap=round, rounded corners, line width=10pt, opacity=0.1] (a0) -- (q0) -- (p1);
        
    \end{tikzpicture}
    \caption{Schematic depiction of the interactions between~$G_i$ and~$G_{i-1}$ during
    the morning and afternoon phases of~$G_i$.}
    \label{fig:morning_afternoon}
\end{figure}

\paragraph*{Morning.} We start with~$G_i$ in the morning state, and~$G_{i-1}$ asleep.
To wake up~$G_{i-1}$, the point~$p_i$ moves to~$\C_1(G_{i-1})$, which currently
contains~$a_{i-1}$ and~$q_{i-1}$. This triggers the wakeup phase of~$G_{i-1}$;
we will analyze this phase later from the perspective of~$G_i$. When
the wakeup phase completes,~$\C_1(G_{i-1})$ contains~$a_{i-1}$ and~$p_i$, and~$p_i$ moves to~$\C_1(G_i)$.
Subsequently~$q_i$ moves from~$\C_0(G_i)$ to~$\C_1(G_i)$. Observe
that this puts~$G_i$ into the afternoon state.

\paragraph*{Afternoon.} In this state,~$G_i$ is once again watching~$G_{i-1}$. Once
the latter falls asleep,~$p_i$ moves from~$\C_1(G_i)$ to~$\C_1(G_{i-1})$, which
triggers another wakeup phase of~$G_{i-1}$. Additionally, this move causes~$G_i$
to fall asleep. Thus, at the end of the wakeup phase of~$G_{i-1}$, we have~$G_{i+1}$ wake up~$G_i$.

\paragraph*{Waking up.} First, the point~$p_{i+1}$ joins~$\C_1(G_i)$. Next,~$p_i$ moves from~$\C_1(G_{i-1})$ to~$\C_0(G_{i})$.  Then,~$q_i$ moves
from~$\C_1(G_i)$ to~$\C_0(G_i)$, and finally,~$p_{i+1}$ leaves~$\C_1(G_i)$, and joins either~$\C_1(G_{i+1})$ (if~$G_{i+1}$ was in
the morning state when waking up~$G_i$) or~$\C_0(G_{i+1})$
(if~$G_{i+1}$ was in the afternoon state; in this case, the move
of~$p_{i+1}$ occurs during the wakeup phase of~$G_{i+1}$).

\paragraph*{Leaf gadget.} The leaf gadget does not watch or wake up
any other gadgets. It only wakes up when~$p_1$ moves into~$\C_0(G_0)$, and falls asleep again when~$p_1$ moves back to a cluster
of~$G_1$.

\paragraph*{Initialization.} The sequence starts with all gadgets in the
asleep, except for~$G_{m-1}$, which is in its morning state.

\vspace{1em}

At every step,
we have the gadget with the smallest index that is not asleep wake up the gadget
that it is watching.
From this sequence of iterations, we can retrieve a series of inequalities,
each of which encodes the condition that the gain of every iteration must be positive.
To prove \Cref{thm:lower_bound}, we must show that the points
in \Cref{table:unit_gadget_points} satisfy these inequalities.

An implementation of the sequence described above is provided in the following link:
\url{https://pastebin.com/raw/McdArCWg}.

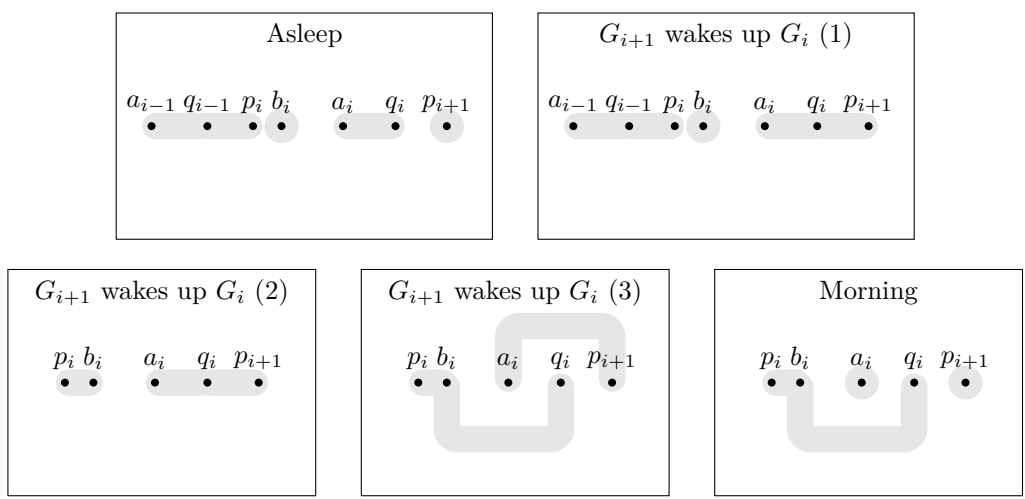
\begin{figure}
    \centering
    \begin{tikzpicture}[scale=1.5]
        \node[] at (1.25, 0.8) {Asleep};
        \draw (-0.4, -1) -- (-0.4, 1) -- (2.9, 1) -- (2.9, -1) -- (-0.4, -1); 
        \node[dot, label=above:$a_i$] (a1) at (1.59, 0) {};
        \node[dot, label=above:$b_i$] (b1) at (1.05, 0) {};
        \node[dot, label=above:$p_i$] (p1) at (0.8, 0) {};
        \node[dot, label=above:$q_i$] (q1) at (2.05, 0) {};
        
        \node[dot, label=above:$a_{i-1}$] (a0) at (-0.09, 0) {};
        \node[dot, label=above:$q_{i-1}$] (q0) at (0.4, 0) {};

        \node[dot, label=above:$p_{i+1}$] (p2) at (2.5, 0) {};

        \coordinate (pbelow) at (0.8, -0.5) {};
        \coordinate (bbelow) at (1.05, -0.5) {};
        \coordinate (qbelow) at (2.05, -0.5) {};
        \coordinate (abelow) at (1.59, -0.5) {};
        
        \coordinate (pabove) at (0.8, 0.5) {};
        \coordinate (aabove) at (1.59, 0.5) {};

        \node[ellipse, draw, fill, opacity=0.1, rotate fit = 45, fit=(b1)] {};
        
        \draw[line cap=round, rounded corners, line width=10pt, opacity=0.1] (q1) -- (a1);
        
        \draw[line cap=round, rounded corners, line width=10pt, opacity=0.1] (a0) -- (q0) -- (p1);
        
        \node[ellipse, draw, fill, opacity=0.1, rotate fit = 45, fit=(p2)] {};

    \end{tikzpicture}
    \hspace{1em}
    \begin{tikzpicture}[scale=1.5]
        \node[] at (1.25, 0.8) {$G_{i+1}$ wakes up~$G_i$ (1)};
        \draw (-0.4, -1) -- (-0.4, 1) -- (2.9, 1) -- (2.9, -1) -- (-0.4, -1); 
        \node[dot, label=above:$a_i$] (a1) at (1.59, 0) {};
        \node[dot, label=above:$b_i$] (b1) at (1.05, 0) {};
        \node[dot, label=above:$p_i$] (p1) at (0.8, 0) {};
        \node[dot, label=above:$q_i$] (q1) at (2.05, 0) {};
        
        \node[dot, label=above:$a_{i-1}$] (a0) at (-0.09, 0) {};
        \node[dot, label=above:$q_{i-1}$] (q0) at (0.4, 0) {};

        \node[dot, label=above:$p_{i+1}$] (p2) at (2.5, 0) {};

        \coordinate (pbelow) at (0.8, -0.5) {};
        \coordinate (bbelow) at (1.05, -0.5) {};
        \coordinate (qbelow) at (2.05, -0.5) {};
        \coordinate (abelow) at (1.59, -0.5) {};
        
        \coordinate (pabove) at (0.8, 0.5) {};
        \coordinate (aabove) at (1.59, 0.5) {};

        \node[ellipse, draw, fill, opacity=0.1, rotate fit = 45, fit=(b1)] {};
        
        \draw[line cap=round, rounded corners, line width=10pt, opacity=0.1] (p2) -- (q1) -- (a1);
        
        \draw[line cap=round, rounded corners, line width=10pt, opacity=0.1] (a0) -- (q0) -- (p1);
        
    \end{tikzpicture}
    \\\vspace{1em}
    \begin{tikzpicture}[scale=1.5]
        \node[] at (1.65, 0.8) {$G_{i+1}$ wakes up~$G_i$ (2)};
        \draw (0.3, -1) -- (0.3, 1) -- (3.0, 1) -- (3.0, -1) -- (0.3, -1); 
        \node[dot, label=above:$a_i$] (a1) at (1.59, 0) {};
        \node[dot, label=above:$b_i$] (b1) at (1.05, 0) {};
        \node[dot, label=above:$p_i$] (p1) at (0.8, 0) {};
        \node[dot, label=above:$q_i$] (q1) at (2.05, 0) {};
        
        \node[dot, label=above:$p_{i+1}$] (p2) at (2.5, 0) {};

        \coordinate (pbelow) at (0.8, -0.5) {};
        \coordinate (bbelow) at (1.05, -0.5) {};
        \coordinate (qbelow) at (2.05, -0.5) {};
        \coordinate (abelow) at (1.59, -0.5) {};
        
        \coordinate (pabove) at (0.8, 0.5) {};
        \coordinate (aabove) at (1.59, 0.5) {};

        \draw[line cap=round, rounded corners, line width=10pt, opacity=0.1] (p1) -- (b1);
        
        \draw[line cap=round, rounded corners, line width=10pt, opacity=0.1] (a1) -- (q1) -- (p2);
        
    \end{tikzpicture}
    \hspace{1em}
    \begin{tikzpicture}[scale=1.5]
        \node[] at (1.65, 0.8) {$G_{i+1}$ wakes up~$G_i$ (3)};
        \draw (0.3, -1) -- (0.3, 1) -- (3.0, 1) -- (3.0, -1) -- (0.3, -1); 
        \node[dot, label=above:$a_i$] (a1) at (1.59, 0) {};
        \node[dot, label=above:$b_i$] (b1) at (1.05, 0) {};
        \node[dot, label=above:$p_i$] (p1) at (0.8, 0) {};
        \node[dot, label=above:$q_i$] (q1) at (2.05, 0) {};
        
        \node[dot, label=above:$p_{i+1}$] (p2) at (2.5, 0) {};

        \coordinate (pbelow) at (0.8, -0.5) {};
        \coordinate (bbelow) at (1.05, -0.5) {};
        \coordinate (qbelow) at (2.05, -0.5) {};
        \coordinate (abelow) at (1.59, -0.5) {};
        
        \coordinate (pabove) at (0.8, 0.5) {};
        \coordinate (p2above) at (2.5, 0.5) {};
        \coordinate (aabove) at (1.59, 0.5) {};

        \draw[line cap=round, rounded corners, line width=10pt, opacity=0.1] (p1) -- (b1) -- (bbelow) -- (qbelow) -- (q1);
        
        \draw[line cap=round, rounded corners, line width=10pt, opacity=0.1] (a1) -- (aabove) -- (p2above) -- (p2);
        
    \end{tikzpicture}
    \hspace{1em}
    \begin{tikzpicture}[scale=1.5]
        \node[] at (1.65, 0.8) {Morning};
        \draw (0.3, -1) -- (0.3, 1) -- (3.0, 1) -- (3.0, -1) -- (0.3, -1); 
        \node[dot, label=above:$a_i$] (a1) at (1.59, 0) {};
        \node[dot, label=above:$b_i$] (b1) at (1.05, 0) {};
        \node[dot, label=above:$p_i$] (p1) at (0.8, 0) {};
        \node[dot, label=above:$q_i$] (q1) at (2.05, 0) {};
        
        \node[dot, label=above:$p_{i+1}$] (p2) at (2.5, 0) {};

        \coordinate (pbelow) at (0.8, -0.5) {};
        \coordinate (bbelow) at (1.05, -0.5) {};
        \coordinate (qbelow) at (2.05, -0.5) {};
        \coordinate (abelow) at (1.59, -0.5) {};
        
        \coordinate (pabove) at (0.8, 0.5) {};
        \coordinate (p2above) at (2.5, 0.5) {};
        \coordinate (aabove) at (1.59, 0.5) {};

        \draw[line cap=round, rounded corners, line width=10pt, opacity=0.1] (p1) -- (b1) -- (bbelow) -- (qbelow) -- (q1);
        
        \node[ellipse, draw, fill, opacity=0.1, rotate fit = 45, fit=(a1)] {};
        
        \node[ellipse, draw, fill, opacity=0.1, rotate fit = 45, fit=(p2)] {};
        
    \end{tikzpicture}
    \caption{Schematic depiction of the interactions between~$G_i$,~$G_{i-1}$  and~$G_{i+1}$ during the
    wakeup phase of~$G_i$. Note that the final state of~$G_i$ corresponds to the first
    state depicted in \Cref{fig:morning_afternoon}.}
    \label{fig:wakeup}
\end{figure}

\section{Smoothed Analysis}\label{sec:smoothed}

For a smoothed analysis, the first hope might be to straightforwardly adapt a smoothed analysis of
Lloyd's algorithm, e.g.\ that of Arthur, Manthey and R\"oglin \cite{arthurSmoothedAnalysisKMeans2011}. On closer inspection, however,
such analyses strongly rely on a couple of properties of Lloyd's method that are not
valid in the Hartigan--Wong method.

First, in Lloyd's algorithm the hyperplane that bisects two cluster centers
also separates their corresponding clusters, since every point is always assigned
to the cluster center closest to itself. Second, the two stages of Lloyd's algorithm,
moving the cluster centers and reassigning points, both decrease the potential.
Neither of these properties are satisfied by iterations of the Hartigan--Wong method. Hence, any
analysis that relies on either property cannot be easily repurposed. 

Instead, we will use a different technique, more closely related to the analysis
of the Flip heuristic for Max-Cut with squared Euclidean distances by
Etscheid and R\"oglin \cite{etscheidSmoothedAnalysisLocal2017}. The main result we will work
towards in this section is stated in \Cref{thm:smoothed_complexity}.

\subsection{Technical Preliminaries}

Let~$\Y \subseteq [0, 1]^d$ be a set of~$n$ points. Throughout the remainder,
we will denote by~$\X$ the set of points obtained by perturbing each point
in~$\Y$ independently by a~$d$-dimensional Gaussian vector of mean~$0$
and standard deviation~$\sigma \leq 1$. Note that this last assumption
is not actually a restriction. If~$\sigma > 1$, we scale down the
set~$\Y$ so that~$\Y \subseteq [0, 1/\sigma]^d$, and subsequently
perturb the points by Gaussian variables with~$\sigma = 1$. Since the
number of iterations required to terminate is invariant under scaling of the
input point set, this is equivalent to the original instance.

Our analysis is based on the standard technique of proving that it is unlikely
that a sequence of iterations decreases the potential function by a small amount. For
this technique to work, we additionally require the potential function to be bounded from above
and from below
with
sufficiently high probability. Since it is obvious that the potential is non-negative
for any clustering, it is enough to guarantee that the perturbed
point set~$\X$ lies within the hypercube~$[-D/2, D/2]^d$
for some finite~$D$.
To that end, we have the following lemma.

\begin{lemma}\label{lemma:points_bounded}
    Let~$D = \sqrt{2n \ln (nkd)}$.
    Then~$\prob(\X \nsubseteq [-D/2, D/2]^d) \leq k^{-n}$.
\end{lemma}

Similar results to \Cref{lemma:points_bounded} can be found in previous works on
the smoothed analysis of algorithms on Gaussian-perturbed point sets
\cite{mantheySmoothedAnalysis2Opt2013,arthurSmoothedAnalysisKMeans2011}. The only
difference in our version is the value of~$D$. Hence, we omit the proof.

\Cref{lemma:points_bounded} allows us to assume that all points lie within~$[-D/2, D/2]^d$ after the
perturbation. Formally, we must take into account the failure event that any point lies outside
this hypercube. However, since the probability of this event is at most~$k^{-n}$,
this adds only
a negligible~$+1$ to the smoothed complexity bound which we prove
in \Cref{thm:smoothed_complexity}. We therefore
ignore the failure event in the sequel.

We need to show that we can approximate the gain of an iteration if
we have a good approximation to the cluster centers. Recall
that~$\Delta_x(C_i, C_j)$ is the gain of moving a point~$x$
from~$C_i$ to~$C_j$. Since we wish to use approximations to the
centers of~$C_i$ and~$C_j$, it is convenient to define the variable
\[
    \Delta_x^{|C_i|,|C_j|}(a, b) = \frac{|C_i|}{|C_i|-1}\|x - a\|^2 - \frac{|C_j|}{|C_j|+1}
        \|x - b\|^2.
\]
This variable is the gain that would be incurred if the centers of~$C_i$ and~$C_j$, with fixed sizes~$|C_i|$ and~$|C_j|$, were~$a$ and~$b$. Indeed, note that~$\Delta_x^{|C_i|,|C_j|}(\cm(C_i), \cm(C_j)) = \Delta_x(C_i, C_j)$.
When their intended values are clear from context, we will often omit the
superscripts~$|C_i|$ and~$|C_j|$ from~$\Delta_x^{|C_i|,|C_j|}(a, b)$.

\subsection{Approximating Iterations}

Before we begin with the analysis, we provide a rough outline.
Suppose we tile the hypercube~$[-D/2, D/2]^d$ with a rectangular grid of spacing~$\epsilon$.
Then any point in~$[-D/2, D/2]^d$ is at a distance of at most~$\sqrt{d}\epsilon$ from
some grid point. 
Since we need the positions of the cluster
centers~$c_i = \cm(C_i)$ for~$i \in [k]$, we guess~$k$ grid points~$c'_i$ for their positions. If we guess correctly, meaning~$c'_i$ is the grid point closest to~$c_i$ for each~$i \in [k]$,
then we can approximate
the gain~$\Delta$ of an iteration by replacing the cluster centers with these grid points in the
formula for~$\Delta$ (\Cref{lemma:approximate_gain}).

The price for this approximation is a union bound over all choices of the grid points.
However, we can compensate for this by noticing that, when we move a point between clusters,
we know exactly how the cluster centers move. Thus, if the guessed grid points are good
approximations, we can obtain new good approximations by moving them the same amount.
Thus, we only need to guess once, and can use this guess for a sequence of iterations.
Then we can bound the probability that all iterations in this sequence
yield a small improvement.

\begin{lemma}\label{lemma:approximate_gain}
    Suppose the point~$x$ moves from cluster~$i$ to cluster~$j$. Let
   ~$C_i$ and~$C_j$ denote the configurations of these clusters before this move,
    and let~$c_i = \cm(C_i)$ and~$c_j = \cm(C_j)$. Let~$c_i'$
    and~$c_j'$ be two points such that
   ~$\|c_i - c_i'\|, \|c_j - c_j'\| \leq \epsilon$ for 
    some~$0 \leq \epsilon \leq \sqrt{d}D$. Then
    \[
        |\Delta_x(C_i, C_j) - \Delta_x(c_i', c_j')| \leq 9\sqrt{d}D \epsilon,
    \]
    In particular,
   ~$\Delta_x(C_i, C_j) \in (0, \epsilon]$ implies 
   ~$|\Delta_x(c_i', c_j')| \leq 10\sqrt{d}D\epsilon$.
\end{lemma}

\begin{proof}
    Observe that
    \[
        \|x - c_i\|^2 = \|x - c_i' + c_i'- c_i\|^2
            = \|x - c_i'\|^2 + \|c_i - c_i'\|^2
                + 2\langle c_i' - c_i, x-c_i'\rangle.
    \]
    Thus,
    \begin{multline*}
        \Delta_x(C_i, C_j) = \Delta_x(c_i', c_j') + \frac{|C_i|}{|C_i|-1}\left(
            \|c_i - c_i'\|^2 + 2\langle c_i' - c_i, x - c_i'\rangle
        \right) \\
        - \frac{|C_j|}{|C_j| + 1} \left(
            \|c_j - c_j'\|^2 + 2\langle c_j' - c_j, x - c_i'\rangle
        \right).
    \end{multline*}
    
    By the Cauchy-Schwarz inequality,
   ~$|\langle c_i' - c_i, x - c_i'\rangle| \leq \epsilon \cdot \|x-c_i'\|$.
    Since all points are contained in~$[-D/2, D/2]^d$, it holds that~$c_i \in [-D/2, D/2]^d$.
    From this fact and the assumption that~$\|c_i - c_i'\| \leq \epsilon \leq \sqrt{d}D$,
    it follows that~$\|x - c_i'\| \leq \sqrt{d}D$.
    
    Moving~$\Delta_x(c_i', c_j')$ to the left and taking an absolute value, we then obtain
    \[
        |\Delta_x(C_i, C_j) - \Delta_x(c_i', c_j')| \leq \left(\frac{|C_i|}{|C_i| - 1}
            + \frac{|C_j|}{|C_j|+1}
        \right) \cdot 3\sqrt{d}D \epsilon.
    \]
    To finish the proof, observe that by \Cref{lemma:never_empty}
    the first term inside the parentheses
    is at most~$2$, while the second term is bounded by~$1$. We then have
    that~$\Delta_x(C_i, C_j) \in (0, \epsilon]$ implies 
   ~$\Delta_x(c_i', c_j') \in (-9\sqrt{d}D\epsilon, (9\sqrt{d}D + 1)\epsilon]$,
    which yields the lemma.
\end{proof}

In the following, we fix a set~$A \subseteq \X$ of active points which will move during
a sequence of the Hartigan--Wong method. 
We also fix the configuration of the active points,
the sizes of the clusters~$|C_1|$ and~$|C_2|$ at the start of the sequence,
and the order~$\pi: A \to [|A|]$ in which the points move.
Observe that these data also fix
the sizes of the clusters whenever a new point moves.

While performing a sequence of iterations, the cluster centers move. Hence,
even if we have a good approximation to a cluster center, it may not remain a good approximation
after the iteration. 
However, if we know which points are gained and lost by each cluster, then
we can compute new good approximations to the cluster centers from the old approximations.
The following lemma captures this intuition.

\begin{lemma}\label{lemma:move_approximation}
    Let~$t_1$,~$t_2$ be two iterations of the Hartigan--Wong method in a sequence
    in which the points~$A \subseteq \X$ move, with~$t_1 < t_2$.
    Suppose in the iterations~$t_1$ through~$t_2-1$, cluster~$i$ loses
    the points~$S_-$ and gains the points~$S_+$. Let~$c_i(t)$ denote the
    cluster center of cluster~$i$ before~$t$ takes place, and let
   ~$C_i^t$ denote its configuration before~$t$.
    Let~$c_i'(t_1) \in \mathbb{R}^d$, and
   ~$c_i'(t_2) = \frac{|C_i^{t_1}|}{|C_i^{t_2}|} c_i'(t_1) + 
    \frac{1}{|C_i^{t_2}|}\left(\sum_{x \in S_+} x
        - \sum_{x \in S_-} x\right)$. Then
    \[
        \|c_i'(t_2) - c_i(t_2)\| = \frac{|C_i^{t_1}|}{|C_i^{t_2}|} \cdot
            \|c_i'(t_1) - c_i(t_1)\|.
    \]
    Moreover, if~$\|c_i'(0) - c_i(0)\| \leq \epsilon$, then
   ~$\|c_i'(t_j) - c_i(t_j)\| \leq 2|A|\epsilon$ for all~$j \in [|A|]$.
\end{lemma}

\begin{proof}
    Since the center of a cluster is defined as its center of mass,
    we can write
    \[
        |C_i^{t_2}|\cm(C_i^{t_2}) = \sum_{x \in C_i^{t_1} \cup S_+ \setminus S_-} x
            = |C_i^{t_1}|\cm(C_i^{t_1}) + \sum_{x \in S_+}  x - \sum_{x \in S_-} x.
    \]
    Thus,
    \[
        |C_i^{t_2}|c_i(t_2) =  |C_i^{t_1}|c_i(t_1) + \sum_{x \in S_+} x - \sum_{x \in S_-} x.
    \]
    Observe then that 
    \[
        \|c_i'(t_2) - c_i(t_2)\| = \frac{|C_i^{t_1}|}{|C_i^{t_2}|} \cdot
            \|c_i'(t_1) - c_i(t_1)\|.
    \]
    This proves the first claim. 
    To prove the second claim, we set~$t_1 = 0$ and~$t_2 = t_j$ for some
   ~$j \in [|A|]$ to obtain
    \[
        \|c_i(t_j) - c_i'(t_j)\| = \frac{|C_i^0|}{|C_i^{t_j}|} \cdot \|c_i(0) - c_i'(0)\|
            \leq (|A| + 1)\epsilon \leq 2|A|\epsilon,
    \]
    since at most~$|A|$ points are active during any subsequence.
\end{proof}

\subsection{Analyzing Sequences}

We now know that we can closely approximate the gain of a sequence of iterations,
provided that we have good approximations to the cluster centers at the
start of the sequence. The next step is then to
show that there is only a small probability that such an approximate sequence improves
the potential by a small amount. For that,
we first require the following technical lemma.

\begin{restatable}{lemma}{affinegaussian}\label{lemma:affine_gaussian}
    Let~$X$ be a~$d$-dimensional Gaussian random variable with arbitrary mean~$\mu$
    and standard deviation~$\sigma \leq 1$,
    and let~$Z = a\|X\|^2 + \langle v, X \rangle$ for fixed
   ~$a \in \mathbb{R} \setminus \{0\}$ and~$v \in \mathbb{R}^d$. Then the probability
    that~$Z$ falls in an interval of size~$\epsilon \leq 1$ is bounded from above by
   ~$O\left(\frac{1}{|a|\sqrt[4]{d}}\sqrt{\frac{\epsilon}{\sigma^2}}\right)$.
\end{restatable}

\begin{proof}
    Let~$Z_i = aX_i^2 + v_i X_i$, so that~$Z = \sum_{i=1}^d Z_i$. We define the auxiliary variable
   ~$\bar{Z}_i = Z_i + v_i^2 / (4a)$ and set~$\bar Z = \sum_{i=1}^d \bar{Z}_i$.
    Since~$a$ and~$v$ are fixed, the densities of~$Z$ and~$\bar{Z}$
    are identical up to translation, and so we can analyze~$\bar{Z}$ instead.
    Observe that~$\bar{Z}_i/a = \left(X_i + \frac{v_i}{2a}\right)^2$.
    Thus,~$\bar{Z}/a$ is equal in distribution to~$\|Y\|^2$, where~$Y$ is a~$d$-dimensional
    Gaussian variable with mean~$\mu + v/(4a)$ and variance~$\sigma^2$.
    We see then that~$\bar{Z}/a$ has the density of a non-central chi-squared distribution.

    For~$\lambda \geq 0$, denote by
   ~$f(x, \lambda, d)$ the non-central~$d$-dimensional chi-squared
    density with non-centrality parameter~$\lambda$ and standard deviation~$\sigma$.
    Then \cite{johnsonContinuousUnivariateDistributions1995}
    \begin{align*}
        f(x, \lambda, d) &= \sum_{i=0}^\infty \frac{e^{-\lambda/2}(\lambda/2)^i}{i!} f(x, 0, d+2i).
    \end{align*}
    Now observe that~$f(x, 0, d)$ is bounded from above by
   ~$O\left(1/(\sqrt{d}\sigma^2)\right)$
    for~$d \geq 2$. We can thus compute for an interval~$I$ of size~$\epsilon$
    \[
        \prob(\|Y\|^2 \in I) = \int_I f(x, \lambda, d) \leq c \cdot \frac{\epsilon}{\sqrt{d}\sigma^2},
    \]
    for some~$c > 0$. Moreover, since probabilities are bounded from above by~$1$,
    we can replace the right-hand side by
    \[
        O\left(\frac{\sqrt{\epsilon}}{\sqrt[4]{d}\sigma}\right).
    \]
    Adding in the scaling factor of~$1/|a|$ then
    yields the lemma for~$d \geq 2$,
    
    For~$d = 1$, we have
    \[
        f(x, 0, 1) = \frac{1}{\sqrt{2\pi\sigma^2}} \cdot \frac{e^{-\frac{x}{2\sigma^2}}}
            {\sqrt{x/\sigma^2}}.
    \]
    Let~$I$ be an interval of size~$\epsilon$. Then
    \begin{align*}
        \prob(\|Y\|^2 \in I) = \int_I f(x, \lambda, 1) \dd x
            \leq \sum_{i=1}^\infty \frac{e^{-\lambda/2}(\lambda/2)^i}{i!} \int_I f(x, 0, 1 + 2i) \dd x
                + \int_I f(x, 0, 1) \dd x.
    \end{align*} 
    The first term is bounded by~$O(\sqrt{\epsilon}/\sigma)$ by the same argument we used for~$d \geq 2$.
    For the second term, we use the expression for~$f(x, 0, 1)$ above to bound the integral as
    \[
        \int_I f(x, 0, 1) \dd x \leq \frac{1}{\sqrt{2\pi\sigma^2}} \int_0^\epsilon
            \frac{e^{-\frac{x}{2\sigma^2}}}{\sqrt{x/\sigma}}\dd x
            = O(\sqrt{\epsilon}/\sigma).
    \]
    This proves the lemma for~$d = 1$ when we again add in the scaling factor~$1/|a|$.
\end{proof}

With \Cref{lemma:affine_gaussian}, we can show that a single fixed approximate iteration
is unlikely to yield a small improvement.

\begin{lemma}\label{lemma:gain_fixed_centers}
    Let~$a, b \in \mathbb{R}^d$ be fixed. Let~$\Delta_x(a, b)$ be the improvement of the first
    move of~$x$ in~$S$, if the cluster centers in this iteration are located at~$a$ and~$b$. Let~$I$
    be an interval of size~$\epsilon \leq 1$. Then
    \[
        \prob(\Delta_x(a, b) \in I) = O\left(
            \frac{n}{\sqrt[4]{d}}\cdot\sqrt{\frac{\epsilon}{\sigma^2}}
        \right).
    \]
\end{lemma}

\begin{proof}
    By \Cref{lemma:move}, we have
    \begin{multline*}
        \Delta_x(a, b) = \frac{|C_i|}{|C_i|-1} \|x - a\|^2 - \frac{|C_j|}{|C_j|+1}\|x - b\|^2 \\
            = \left(\frac{|C_i|}{|C_i|-1} - \frac{|C_j|}{|C_j|+1}\right)\|x\|^2 
                + \left\langle 2\left(
                    \frac{|C_j|}{|C_j|+1} b - \frac{|C_i|}{|C_i|-1} a
                \right),
                    x
                \right\rangle
                    \\ + \frac{|C_i|}{|C_i|-1} \|a\|^2 - \frac{|C_j|}{|C_j|+1}\|b\|^2,
    \end{multline*}
    where~$|C_i|$ and~$|C_j|$ denote the sizes of clusters~$i$ and~$j$ before the iteration,
    and we assume~$x$ moves from cluster~$i$ to cluster~$j$. 

    Since the sizes of the clusters as well as~$a$ and~$b$ are fixed, the last term
    in the above is fixed, and hence we may disregard it when analyzing~$\prob(\Delta_x(a, b) \in I)$.
    Since~$x$ is a Gaussian random variable, we can apply \Cref{lemma:affine_gaussian} to find
    \[
        \prob(\Delta_x(a, b) \in I) = O\left(\left(
            \frac{|C_i|}{|C_i|-1} - \frac{|C_j|}{|C_j|+1}
        \right)^{-1} \cdot 
        \frac{1}{\sqrt[4]{d}} \cdot \sqrt{\frac{\epsilon}{\sigma^2}}\right).
    \]
    It remains to bound quantity in the inner brackets from below.
    Since each cluster is bounded in size by~$n$, we have
    \[
        \frac{|C_i|}{|C_i|-1} - \frac{|C_j|}{|C_j|+1} \geq \frac{n}{n-1} - \frac{n}{n+1} = \frac{2n}{(n-1)(n+1)} \geq \frac{1}{n},
    \]
    and we are done.
\end{proof}

As stated at the start of the analysis, analyzing a single iteration is not enough to
prove \Cref{thm:smoothed_complexity}. The following lemma extends \Cref{lemma:gain_fixed_centers}
to a sequence of iterations, given a fixed point set~$A \subseteq \X$ that moves
in the sequence.

\begin{lemma}\label{lemma:gain_fixed_sequence}
    Fix an active set~$A$ and starting cluster sizes~$|C_i|$ for~$i \in [k]$.
    Moreover, fix an order~$\pi : A \to [|A|]$ in which the points in~$A$
    move, i.e.,~$\pi(x) < \pi(y)$ means~$x$ moves for the first time before
   ~$y$ moves for the first time. Let~$\Delta$ denote the minimum
    improvement of a sequence satisfying these hypotheses over all possible configurations
    of~$\X \setminus A$. Then for~$\epsilon \leq 1$,
    \[
          \prob\left(
            \Delta \leq \epsilon\right) \leq 
            \left(\frac{2D}{\epsilon}\right)^{kd} \cdot \left(\frac{O(1) 
                \cdot k^{|A|}
            \cdot d^{3/4}Dn|A|\sqrt\epsilon}{\sigma}
                \right)^{|A|}.
    \]
\end{lemma}

\begin{proof}
    For~$x \in A$, let~$\Delta_x$ denote the improvement of the first move of~$x \in A$.
    We label the points in~$A$ as~$(x_1, \ldots, x_{|A|})$ according to~$\pi$.
    Let~$\Delta = (\Delta_i)_{i=1}^{|A|}$.

    To compute the vector~$\Delta$, we would need to know the configuration and
    positions of the points~$P = \X \setminus A$, since these are required to compute
    the~$k$ cluster centers. However, if we had approximations to the cluster 
    centers in every iteration
    corresponding to the entries of~$\Delta$, then
    we could compute an approximation to~$\Delta$ by \Cref{lemma:approximate_gain}.

    Since the cluster centers are convex combinations of points in~$[-D/2, D/2]^d$, we know that
    the cluster centers at the start of~$S$ must also lie in~$[-D/2, D/2]^d$. Thus,
    there exist grid points~$c_i'$ ($i \in [k]$) within a distance~$\sqrt{d}\epsilon$ of
    the initial cluster centers. 

    Knowing these grid points, we would like to apply \Cref{lemma:move_approximation}
    in order to update the approximate cluster centers whenever a new point moves.
    We then need to know the points gained and lost by each cluster
    between first moves of each~$x \in A$. Observe that to obtain this information, it suffices
    to know the configuration of the active points before the first move
    of each~$x \in A$. Thus, we fix these configurations.
    
    We collect the gain of each first move of a point in~$A$, where we replace the cluster centers
    by these approximations, into a vector~$\Delta'$.
    By the reasoning above and by \Cref{lemma:approximate_gain,lemma:move_approximation},
    if there exist initial cluster centers~$c_i$ ($i \in [k]$) such that~$\Delta_x \in (0, \epsilon]$
    for all~$x \in A$, then there exist grid points~$c_i'$, such that
   ~$|\Delta_x'| \leq 20|A|dD\epsilon$ for all~$x \in A$. (Compared to
    \Cref{lemma:approximate_gain}, we gain an extra factor of~$2|A|$ due to
    \Cref{lemma:move_approximation}.)

    By this reasoning, it suffices to obtain a bound on 
   ~$\prob\left(\bigcap_{x \in A} |\Delta_x'| \leq 20|A|dD\epsilon\right)$.
    We can then take a union bound over these events for all~$(D/\epsilon + 1)^{kd} \leq (2D/\epsilon)^{kd}$ choices of
   ~$c_i'$ for~$i \in [k]$,
    and a union bound over the configuration of~$A$ before the first move
    of each~$x \in A$.

    To show that~$\prob\left(\bigcap_{x \in A} |\Delta_x'| \leq 20|A|dD\epsilon\right)$ is bounded
    as desired, we consider the following algorithm.
    %
    \begin{enumerate}
        \item Set~$t = 1$.
        \item Reveal~$x_t$, and compute~$\Delta_{x_t}(c_{i_t}', c_{j_t}')$, where~$x_t$ moves
            from~$C_{i_t}$ to~$C_{j_t}$.

        \item If~$|\Delta_{x_t}(c_{i_t}', c_{j_t}')| > 20|A|dD\epsilon$, then return~$\sf{GOOD}$
            and halt.

        \item If~$t = |A|$, return~$\sf{BAD}$.
            
        \item Update the positions of the approximate cluster centers using
            \Cref{lemma:move_approximation}.

        \item Continue executing moves in the sequence until we encounter the first
            move of~$x_{t+1}$. Observe
            that the information we fixed before executing this algorithm suffices
            to compute approximations to the cluster centers whenever a new point
            moves.

        \item Set~$t \leftarrow t + 1$ and go to step 2.
    \end{enumerate}
    The sequence of iterations improves the potential
    by at most~$\epsilon$ only if the above algorithm returns~$\sf{BAD}$. We now argue that 
    \[
        \prob(\mathsf{BAD}) \leq \left(O(1) \cdot
            d^{3/4}Dn|A|\sqrt{\epsilon}/\sigma\right)^{|A|}.
    \]
    Let~$\mathsf{BAD}_t$ be the event that the above algorithm loops for at least~$t$
    iterations. Then~$\prob(\mathsf{BAD}) = \prob(\mathsf{BAD}_{|A|})$.
    Since~$\prob(\mathsf{BAD}_t \given \neg \mathsf{BAD}_{t-1}) = 0$,
    we can immediately conclude that for all~$t \in \{2, \ldots, |A|\}$,
    \[
        \prob(\mathsf{BAD}_t) = \prob(\mathsf{BAD}_t \given \mathsf{BAD}_{t-1})\prob(\mathsf{BAD}_{t-1}).
    \]
    By \Cref{lemma:gain_fixed_centers}, we have
   ~$\prob(\mathsf{BAD}_t \given \mathsf{BAD}_{t-1}) \leq O(1) \cdot 
    d^{3/4}Dn|A|\sqrt{\epsilon}/\sigma$.
    Thus,~$\prob(\mathsf{BAD}_t)$ is bounded as claimed. 
    
    Taking a union bound over all choices of the approximate grid points
    at the start of the sequence yields
    the factor~$(2D/\epsilon)^{kd}$.
    Finally, we must take a union bound over the configuration of~$A$ before
    the first move of each~$x \in A$, yielding a factor~$k^{|A|^2}$,
    which concludes the proof.
\end{proof}

Armed with \Cref{lemma:gain_fixed_sequence}, we can bound the probability that there exists a sequence in which
a fixed number of points moves, which improves the potential by at most~$\epsilon$.

\begin{lemma}\label{lemma:gain_minimum_sequence}
    Let~$\Delta_\mathrm{min}$ denote the minimum improvement of any
    sequence of moves in which exactly~$4kd$ distinct
    points switch clusters. Then for~$\epsilon \leq 1$,
    \[
        \prob(\Delta_\mathrm{min} \leq \epsilon) \leq 
            \left(
            \frac{O(1) \cdot  k^{8kd+4} d^{11}D^5n^{8 + \frac{1}{d}} \epsilon}
            {\sigma^4}
        \right)^{kd}.
    \]
\end{lemma}

\begin{proof}
    Fix an active set~$A$ of~$4kd$ distinct points, an order~$\pi : A \to [|A|]$
    in which the points in~$A$ move, and the sizes of the two clusters at the start of the sequence.
    
    We have
    by \Cref{lemma:gain_fixed_sequence}
    \[
        \prob(\Delta(S) \leq \epsilon) \leq \left(\frac{2D}{\epsilon}\right)^{kd} \left(
            \frac{O(1) \cdot k^{2kd} \cdot d^{7/4}Dn\sqrt\epsilon}
                {\sigma}
        \right)^{4kd}
        = 
        \left(
            \frac{O(1) \cdot d^7\cdot k^{8kd} \cdot D^5n^4\epsilon}
                {\sigma^4}
        \right)^{kd}.
    \]
    We conclude the proof by a union bound over the choices of~$A$,
   ~$\pi$, and the sizes of the clusters at the start of the sequence,
    which yields a factor of at most~$(4kd)^{4kd} \cdot n^{4kd+1}$.
\end{proof}

With \Cref{lemma:gain_minimum_sequence}, we are in a position to prove the
main result of this section. The proof is essentially mechanical,
following techniques used in many previous smoothed analyses 
\cite{arthurSmoothedAnalysisKMeans2011,englertWorstCaseProbabilistic2014, englertSmoothedAnalysis2Opt2016, etscheidSmoothedAnalysisSquared2015, etscheidSmoothedAnalysisLocal2017, mantheySmoothedAnalysis2Opt2013}.

\smoothedcomplexity*

\begin{proof}
    First, we recall that the point set~$\X$ is contained in~$[-D/2, D/2]^d$. This yields
    an upper bound for the value of the potential function for the initial clustering
   ~$C$,
    \[
        \Phi(C) = \sum_{i=1}^k \sum_{x \in C_i}  \|x - \cm(C_i)\|^2
            \leq k n dD^2.
    \]
    We divide the sequence of iterations executed by the Hartigan--Wong method into contiguous disjoint blocks during
    which exactly~$4kd$ distinct points move.
    By \Cref{lemma:gain_minimum_sequence}, we know that the probability
    that any such block yields a bad improvement is small.

    Let~$T$ be the number of such blocks traversed by the heuristic
    before we reach a local optimum. Then
    \[
        \prob\left(T \geq t\right) \leq \prob\left(\Delta_\mathrm{min} \leq \frac{kndD^2}{t}\right)
            \leq \min\left\{
                1, 
                    \frac{O(1) \cdot k^{8kd+5} d^{12} D^7 n^{9+\frac{1}{d}}}
                    {\sigma^4} \cdot \frac{1}{t}
            \right\}.
    \]
    This probability becomes nontrivial when
    \[
        t > \left\lceil
                    \frac{O(1) \cdot k^{8kd+5} d^{12} D^7 n^{9+\frac{1}{d}}}
                    {\sigma^4}
        \right\rceil =: t'.
    \]
    Observe that~$t' = \Omega(kndD^2)$, justifying our use of \Cref{lemma:gain_minimum_sequence} above.
    Thus, we find
    \begin{align*}
        \expect(T) = \sum_{t=1}^{k^n} \prob(T \geq t) \leq t' + t' \cdot \sum_{t=t'}^{k^n} \frac{1}{t}
            \leq t' +  t' \cdot \int_{t'}^{k^n} \frac{1}{t}\dd t
                \leq t' + t' \cdot \ln (k^n).
    \end{align*}
    The upper limit of~$k^n$ to the sum is simply the number of possible clusterings of~$n$
    points into~$k$ sets, which is a trivial upper bound to the number of iterations.
    To conclude, we observe that any block in which exactly~$4kd$ distinct points
    move has a length of at most~$k^{4kd}$, as otherwise some clustering would show up twice.
    Thus, we multiply~$\expect(T)$ by
   ~$k^{4kd}$ to obtain a bound for the smoothed complexity. Finally, we insert the value
    of~$D = \sqrt{2n \ln(nkd)}$.
\end{proof}

\section{Discussion}

\Cref{thm:lower_bound,thm:smoothed_complexity}
provide some of the first rigorous theoretical results concerning the Hartigan--Wong method
method that have been found since Telgarsky \& Vattani explored the heuristic in
2010 \cite{telgarskyHartiganMethodKmeans2010}. Of course, many interesting open questions
still remain.

\subparagraph{Worst-case construction.}
\Cref{thm:lower_bound} establishes the existence of exponential-length sequences
on the line, but leaves open the possibility that
a local optimum may be reachable more efficiently by a different improving sequence.
To be precise: given an instance of~$k$-means clustering on the line
and an initial clustering, does there always exist a sequence of iterations
of the Hartigan--Wong method of length~$\poly(n, k)$ starting from this clustering
and ending in a local optimum? Although the~$d = 1$ case appears very
restricted at first sight, this question seems surprisingly difficult to answer.

In addition, the construction we use in \Cref{thm:lower_bound}
requires~$k = \Theta(n)$ clusters. This opens up the question whether
similar worst-case constructions can be made using fewer, perhaps even~$O(1)$, clusters. Note that this is not true for Lloyd's method, since
the number of iterations of Lloyd's method is bounded by~$n^{O(kd)}$ \cite{inabaVariancebasedKclusteringAlgorithms2000},
which is polynomial for~$k, d \in O(1)$.

\subparagraph{Smoothed complexity.}
\Cref{thm:smoothed_complexity} entails, to our knowledge, the first
step towards settling the conjecture by Telgarsky \& Vattani \cite{telgarskyHartiganMethodKmeans2010}
that
the Hartigan--Wong method has polynomial smoothed complexity.
Our result is reminiscent of the smoothed
complexity bound of Lloyd's method obtained in 2009 by Manthey \& R\"oglin \cite{mantheyImprovedSmoothedAnalysis2009},
which is~$k^{kd}\cdot \poly(n, 1/\sigma)$. In the case of
Lloyd's method, the smoothed complexity was later settled to~$\poly(n, k, d, 1/\sigma)$ \cite{arthurSmoothedAnalysisKMeans2011}. 

Observe that our bound is polynomial for constant~$k$ and~$d$, and even
for~$kd\log k \in O(\log n)$. While this is certainly an improvement
over the trivial upper bound of~$k^n$, it falls short of a true polynomial
bound.
We hope that our result can function
as a first step to a~$\poly(n, k, d, 1/\sigma)$ smoothed complexity
bound of the Hartigan--Wong method.

We remark that the exponents in the bound in \Cref{thm:smoothed_complexity} can be
easily improved by a constant factor
for~$d \geq 2$. The reason is that in \Cref{lemma:affine_gaussian}, the factor~$\sqrt{\epsilon}$ emerges from the~$d = 1$ case, while for~$d \geq 2$ we could
instead obtain~$\epsilon$.
We chose to combine these cases for the sake of keeping the analysis simple,
as we expect the bound in \Cref{thm:smoothed_complexity} would be far from optimal
regardless.

\subparagraph{Improving the smoothed bound.}
We do not believe that the factor of~$k^{O(kd)}$ is inherent in
the smoothed complexity of the Hartigan--Wong method, but is rather
an artifact of our analysis. To replace this factor by a polynomial
in~$k$ and~$d$,
it seems that significantly new ideas might be needed.

The factors arise from two sources in our analysis. First, we take
a union bound over the configuration of the active points each time we
apply \Cref{lemma:move_approximation}, yielding factors of~$k^{O(kd)}$.
Second, we analyze sequences in which~$\Theta(kd)$ points move
in order to guarantee a significant potential decrease. This incurs
a factor of the length of such a sequence, which is another source of a factor~$k^{O(kd)}$.
We do not see how to avoid such factors when taking
our approach.

One avenue for resolving this problem might be to analyze shorter sequences
in which a significant number of points move. Angel et al.\ used such an approach
in their analysis of the Flip heuristic for Max-Cut.
They identify in any sequence~$L$ of moves a shorter subsequence~$B$, such that
the number of unique vertices that flip in~$B$ is linear in the length of~$B$.
The major challenge is then to find sufficient independence in such a short subsequence,
which in our case seems challenging, as we need to compensate for a factor~$\epsilon^{-kd}$ in \Cref{lemma:gain_fixed_sequence}.

Since our analysis greatly resembles the earlier analysis
of the Flip heuristic for Squared Euclidean Max Cut \cite{etscheidSmoothedAnalysisLocal2017},
it might be helpful to first improve the latter. This analysis
yields a bound of~$2^{O(d)}\cdot\poly(n, 1/\sigma)$. If this can
be improved to~$\poly(n, d, 1/\sigma)$, then it is likely that
a similar method can improve on our analysis for the Hartigan--Wong method
as well.

\bibliography{bibliography.bib}

\begin{thebibliography}{10}

\bibitem{arthurSmoothedAnalysisKMeans2011}
David Arthur, Bodo Manthey, and Heiko R{\"o}glin.
\newblock Smoothed {{Analysis}} of the k-{{Means Method}}.
\newblock {\em Journal of the ACM}, 58(5):19:1--19:31, October 2011.
\newblock \href {https://doi.org/10.1145/2027216.2027217}
  {\path{doi:10.1145/2027216.2027217}}.

\bibitem{arthurKmeansAdvantagesCareful2007}
David Arthur and Sergei Vassilvitskii.
\newblock K-means++: The advantages of careful seeding.
\newblock In {\em Proceedings of the Eighteenth Annual {{ACM-SIAM}} Symposium
  on {{Discrete}} Algorithms}, {{SODA}} '07, pages 1027--1035, {USA}, January
  2007. {Society for Industrial and Applied Mathematics}.

\bibitem{englertWorstCaseProbabilistic2014}
Matthias Englert, Heiko R{\"o}glin, and Berthold V{\"o}cking.
\newblock Worst {{Case}} and {{Probabilistic Analysis}} of the 2-{{Opt
  Algorithm}} for the {{TSP}}.
\newblock {\em Algorithmica}, 68(1):190--264, January 2014.
\newblock \href {https://doi.org/10.1007/s00453-013-9801-4}
  {\path{doi:10.1007/s00453-013-9801-4}}.

\bibitem{englertSmoothedAnalysis2Opt2016}
Matthias Englert, Heiko R{\"o}glin, and Berthold V{\"o}cking.
\newblock Smoothed {{Analysis}} of the 2-{{Opt Algorithm}} for the {{General
  TSP}}.
\newblock {\em ACM Transactions on Algorithms}, 13(1):10:1--10:15, September
  2016.
\newblock \href {https://doi.org/10.1145/2972953} {\path{doi:10.1145/2972953}}.

\bibitem{etscheidSmoothedAnalysisSquared2015}
Michael Etscheid and Heiko R{\"o}glin.
\newblock Smoothed {{Analysis}} of the {{Squared Euclidean Maximum-Cut
  Problem}}.
\newblock In Nikhil Bansal and Irene Finocchi, editors, {\em Algorithms -
  {{ESA}} 2015}, Lecture {{Notes}} in {{Computer Science}}, pages 509--520,
  {Berlin, Heidelberg}, 2015. {Springer}.
\newblock \href {https://doi.org/10.1007/978-3-662-48350-3_43}
  {\path{doi:10.1007/978-3-662-48350-3_43}}.

\bibitem{etscheidSmoothedAnalysisLocal2017}
Michael Etscheid and Heiko R{\"o}glin.
\newblock Smoothed {{Analysis}} of {{Local Search}} for the {{Maximum-Cut
  Problem}}.
\newblock {\em ACM Transactions on Algorithms}, 13(2):25:1--25:12, March 2017.
\newblock \href {https://doi.org/10.1145/3011870} {\path{doi:10.1145/3011870}}.

\bibitem{gurobioptimizationllcGurobiOptimizerReference2023}
{Gurobi Optimization LLC}.
\newblock Gurobi {{Optimizer Reference Manual}}.
\newblock Gurobi Optimization, LLC, 2023.

\bibitem{hartiganAlgorithm136KMeans1979}
J.~A. Hartigan and M.~A. Wong.
\newblock Algorithm {{AS}} 136: {{A K-Means Clustering Algorithm}}.
\newblock {\em Journal of the Royal Statistical Society. Series C (Applied
  Statistics)}, 28(1):100--108, 1979.
\newblock \href {http://arxiv.org/abs/2346830} {\path{arXiv:2346830}}, \href
  {https://doi.org/10.2307/2346830} {\path{doi:10.2307/2346830}}.

\bibitem{inabaVariancebasedKclusteringAlgorithms2000}
M~Inaba, Naoki Katoh, and Hiroshi Imai.
\newblock Variance-based k-clustering algorithms by {{Voronoi}} diagrams and
  randomization.
\newblock {\em IEICE Transactions on Information and Systems}, E83D, June 2000.

\bibitem{johnsonContinuousUnivariateDistributions1995}
Norman~L. Johnson, Samuel Kotz, and Narayanaswamy Balakrishnan.
\newblock {\em Continuous {{Univariate Distributions}}, {{Volume}} 2}.
\newblock {John Wiley \& Sons}, May 1995.

\bibitem{lloydLeastSquaresQuantization1982}
S.~Lloyd.
\newblock Least squares quantization in {{PCM}}.
\newblock {\em IEEE Transactions on Information Theory}, 28(2):129--137, March
  1982.
\newblock \href {https://doi.org/10.1109/TIT.1982.1056489}
  {\path{doi:10.1109/TIT.1982.1056489}}.

\bibitem{mantheyImprovedSmoothedAnalysis2009}
Bodo Manthey and Heiko R{\"o}glin.
\newblock Improved {{Smoothed Analysis}} of the k-{{Means Method}}.
\newblock In {\em Proceedings of the 2009 {{Annual ACM-SIAM Symposium}} on
  {{Discrete Algorithms}} ({{SODA}})}, Proceedings, pages 461--470. {Society
  for Industrial and Applied Mathematics}, January 2009.
\newblock \href {https://doi.org/10.1137/1.9781611973068.51}
  {\path{doi:10.1137/1.9781611973068.51}}.

\bibitem{mantheySmoothedAnalysis2Opt2013}
Bodo Manthey and Rianne Veenstra.
\newblock Smoothed {{Analysis}} of the 2-{{Opt Heuristic}} for the {{TSP}}:
  {{Polynomial Bounds}} for {{Gaussian Noise}}.
\newblock In Leizhen Cai, Siu-Wing Cheng, and Tak-Wah Lam, editors, {\em
  Algorithms and {{Computation}}}, Lecture {{Notes}} in {{Computer Science}},
  pages 579--589, {Berlin, Heidelberg}, 2013. {Springer}.
\newblock \href {https://doi.org/10.1007/978-3-642-45030-3_54}
  {\path{doi:10.1007/978-3-642-45030-3_54}}.

\bibitem{telgarskyHartiganMethodKmeans2010}
Matus Telgarsky and Andrea Vattani.
\newblock Hartigan's {{Method}}: K-means {{Clustering}} without {{Voronoi}}.
\newblock In {\em Proceedings of the {{Thirteenth International Conference}} on
  {{Artificial Intelligence}} and {{Statistics}}}, pages 820--827. {JMLR
  Workshop and Conference Proceedings}, March 2010.

\bibitem{vattaniKmeansRequiresExponentially2011}
Andrea Vattani.
\newblock K-means {{Requires Exponentially Many Iterations Even}} in the
  {{Plane}}.
\newblock {\em Discrete \& Computational Geometry}, 45(4):596--616, June 2011.
\newblock \href {https://doi.org/10.1007/s00454-011-9340-1}
  {\path{doi:10.1007/s00454-011-9340-1}}.

\end{thebibliography}

\iflong

\clearpage

\appendix

\section{Proof of Theorem \ref{thm:lower_bound}}\label{sec:lower_bound_proof}

\lowerbound*

\begin{proof}
Inspecting the sequence detailed in \Cref{sec:lower_bound_sequence}, we observe that
the morning phase consists of three iterations, the afternoon phase of
one iteration, and the waking up phase of six iterations. 
Each iteration yields a single inequality, and thus we must verify ten inequalities
in total.

Recall the notation
\[
    \Delta_x(S, T) = \frac{|S|}{|S|-1} \|x - \cm(S)\|^2 - \frac{|T|}{|T|+1} \|x - \cm(T)\|^2,
\]
which is the gain obtained by moving the point~$x$ from~$S$ to~$T$. Then
each iteration will yield an inequality of the form~$\Delta_x(S, T) > 0$.

Note that we may assume w.l.o.g.\ that~$G_i$ is the unit gadget, and as such that for each triple of corresponding points~$x_{i-1}, x_i, x_{i+1}$ it holds that~$x_{i-1} = (x_i - t_0) / 5$ and~$x_{i+1} = 5x_i + t_0$.
\paragraph*{Morning.} This phase yields the inequalities:
\begin{align}
    \Delta_{p_i}(\{p_i, q_i, b_i\}, \{a_{i-1}, q_{i-1}\}) > 0, \label[ineq]{ineq:morning1} \\
    \Delta_{p_i}(\{a_{i-1}, p_i\} , \{a_i\}) > 0, \label[ineq]{ineq:morning2} \\
    \Delta_{q_i}(\{b_i, q_i\}, \{a_i, p_i\}) > 0. \label[ineq]{ineq:morning3}
\end{align}
For the second of these inequalities, we refer to the description of the wakeup
phase, which shows that~$\C_1(G_{i-1})$ takes the value~$\{a_{i-1}, p_i\}$ at the
moment~$p_i$ moves back to a cluster of~$G_i$.
\paragraph*{Afternoon.} This phase only yields
\begin{align}\label[ineq]{ineq:afternoon}
    \Delta_{p_i}(\{a_i, p_i, q_i\}, \{a_{i-1}, q_{i-1}\}) > 0.
\end{align}
~\paragraph*{Waking up.} Here, we have
\begin{align*}
    \Delta_{p_{i+1}}(\{p_{i+1}, q_{i+1}, b_{i+1}\}, \{a_i, q_i\}) > 0, \label[ineq]{ineq:first_wc}\\
    \Delta_{p_{i+1}}(\{p_{i+1}, q_{i+1}, a_{i+1}\}, \{a_i, q_i\}) > 0, \\
    \Delta_{p_i}(\{a_{i-1}, p_i\}, \{b_i\}) > 0, \\
    \Delta_{q_i}(\{a_i, q_i, p_{i+1}\}, \{b_i\}) > 0, \\
    \Delta_{p_{i+1}}(\{a_i, p_{i+1}\}, \{a_{i+1}\}) > 0, \\
    \Delta_{p_{i+1}}(\{a_i, p_{i+1}\}, \{b_{i+1}\}) > 0.
\end{align*}
Observe that some of these inequalities are equivalent to one another, or
to the inequalities specified in the morning and afternoon phases.
The only unique inequalities are
\begin{align}
    \Delta_{p_i}(\{a_{i-1}, p_i\}, \{b_i\}) > 0, \\
    \Delta_{q_i}(\{a_i, q_i, p_{i+1}\}, \{p_i, b_i\}) > 0.
\end{align}
\paragraph*{Leaf gadget.} Finally, we consider the case where~$G_{i-1}$ is the leaf gadget.
This yields the additional inequalities
\begin{align}
    \Delta_{p_i}(\{p_i, q_i, b_i\}, \{f\}) > 0,  \\
    \Delta_{p_i}(\{f, p_i\} , \{a_i\}) > 0, \\
    \Delta_{p_i}(\{a_i, p_i, q_i\}, \{f\}) > 0, \\
    \Delta_{p_i}(\{f, p_i\}, \{b_i\}) > 0. \label[ineq]{ineq:last_wc}
\end{align}

\vspace{1em}

Verifying the inequalities is now simply a matter of plugging in the values from
\Cref{table:unit_gadget_points}. 

\paragraph*{Morning.}
We start by computing some necessary quantities:
\begin{align*}
  q_{i-1} &= (q_i - t_0)/5 = (13 - 8) / 5 = 1 , \\
  a_{i-1} &= (a_i - t_0)/5 = (9 - 8) / 5 = 0.2, \\
  \cm(\{p_i, q_i, b_i\}) &= \frac{1}{3}(5 + 13 + 6) / 3 = 8 \\
  \cm(\{a_{i-1}, q_{i-1}\}) &= \frac{1}{2}(0.2 + 1) = 0.6, \\
  \cm(\{a_{i-1}, p_i\}) &= \frac{1}{2}(0.2 + 5) = 2.6, \\
  \cm(\{a_i\}) &= a_i = 9, \\
  \cm(\{b_i, q_i\}) &= \frac{1}{2} (6 + 13) = 9.5, \\
  \cm(\{a_i, p_i\}) &= \frac{1}{2}(9 + 5) = 7.
\end{align*}
Using these quantities, we directly compute:
\begin{align*}
        \Delta_{p_i}(\{p_i, q_i, b_i\}, \{a_{i-1}, q_{i-1}\}) &= \frac{3}{2}\left(  5 - 8 \right)^2
            - \frac{2}{3}\left(
                5 - 0.6
            \right)^2 \approx 0.5933 > 0, \\
    \Delta_{p_i}(\{a_{i-1}, p_i\}, \{a_i\}) &= 2 \cdot \left(
        5 - 2.6
    \right)^2
        - \frac{1}{2} \cdot \left(
        5 - 9
    \right)^2
    = 3.52 > 0, \\
    \Delta_{q_i}(\{b_i, q_i\}, \{a_i, p_i\}) &= 2 \cdot \left(
        13 - 9.5
    \right)^2
      -\frac{1}{2} \cdot \left(
        13 - 7
    \right)^2
    = 6.5 > 0.
\end{align*}
Thus, \Cref{ineq:morning1,ineq:morning2,ineq:morning3} are satisfied by the
points given in \Cref{table:unit_gadget_points}. 
\paragraph*{Afternoon.} The inequality we must verify for this phase is
\begin{align*}
    \Delta_{p_i}(\{a_i, p_i, q_i\}, \{a_{i-1}, q_{i-1}\}) > 0.
\end{align*}
We already know that~$\cm(\{a_{i-1}, q_{i-1}) = 0.6$.
Moreover,
\[
    \cm(\{a_i, p_i, q_i\}) = \frac{1}{3}(9 + 5 + 13) = 9.
\]
Thus,
\[
    \Delta_{p_i}(\{a_i, p_i, q_i\}, \{a_{i-1}, q_{i-1}\})
        = \frac{3}{2} (5 - 9)^2 - \frac{2}{3} (5 - 0.6)^2 \approx 11.093 > 0
\]
as claimed.
\paragraph*{Waking up.}  We must verify 
\begin{align*}
    \Delta_{p_i}(\{a_{i-1}, p_i\}, \{b_i\}) > 0, \\
    \Delta_{q_i}(\{a_i, q_i, p_{i+1}\}, \{p_i, b_i\}) > 0.
\end{align*}
We already know that~$\cm(\{a_{i-1}, p_i\}) = 2.6$. Thus,
\[
    \Delta_{p_i}(\{a_{i-1}, p_i\}, \{b_i\}) = 2 \cdot (5 - 2.6)^2 - 0.5 \cdot (5 - 6)^2 
        = 11.02 > 0.
\]
Next,
we have~$p_{i+1} = 5p_i + t_0 = 33$. Thus,
\begin{align*}
    \cm(\{a_i, q_i, p_{i+1}\}) &= \frac{1}{3}(9 + 13 + 33) = 55/3, \\
    \cm(\{p_i, b_i\}) &= \frac{1}{2} (5 + 6) = 5.5,
\end{align*}
enabling us to verify that
\[
    \Delta_{q_i}(\{a_i, q_i, p_{i+1}\}, \{p_i, b_i\}) = 
        \frac{3}{2} \cdot (13 - 55/3)^2  - \frac{2}{3} \cdot (13 - 5.5)^2 \approx 5.167 > 0.
\]
\paragraph*{Leaf gadget.} We must verify
\begin{align*}
    \Delta_{p_i}(\{p_i, q_i, b_i\}, \{f\}) > 0,  \\
    \Delta_{p_i}(\{f, p_i\} , \{a_i\}) > 0, \\
    \Delta_{p_i}(\{a_i, p_i, q_i\}, \{f\}) > 0, \\
    \Delta_{p_i}(\{f, p_i\}, \{b_i\}) > 0. \label{ineq:last_wc}
\end{align*}
We first compute
\begin{align*}
    \cm(\{p_i, q_i, b_i\}) &= \frac{1}{3}(5+13+6) = 8, \\
    \cm(\{f, p_i\}) &= \frac{1}{2} (0 + 5) = 2.5.
\end{align*}
Finally,
\begin{align*}
    \Delta_{p_i}(\{p_i, q_i, b_i\}, \{f\})
        &= \frac{3}{2}(5 - 8)^2 - 0.5\cdot (5 - 0)^2 = 1 > 0, \\
    \Delta_{p_i}(\{f, p_i\}, \{a_i\}) &= 2 \cdot (5 - 2.5)^2 - 0.5 \cdot (5 - 9)^2
        = 4.5 > 0, \\
    \Delta_{p_i}(\{a_i, p_i, q_i\}, \{f\}) &= \frac{3}{2} \cdot (5 - 9)^2
        - 0.5 \cdot (5 - 0)^2 = 11.5 > 0, \\
    \Delta_{p_i}(\{f, p_i\}, \{b_i\}) &= 2 \cdot (5 - 2.5)^2 - 0.5 \cdot (5 - 6)^2
        = 12 > 0.
\end{align*}

This
concludes the verification of \Crefrange{ineq:morning1}{ineq:last_wc} and
therefore the proof of \Cref{thm:lower_bound}.
\end{proof}

\fi

\end{document}